\documentclass[a4paper,english,11pt]{amsart}
\pdfoutput=1
\usepackage{amsthm}
\usepackage{graphicx} \usepackage{array} 

\usepackage{amsmath, amssymb, amsfonts, verbatim}
\usepackage{hyphenat, epsfig, subcaption, multirow}
\usepackage{nicefrac}
\usepackage{paralist}

\usepackage{mathtools, etoolbox}
\usepackage{thmtools}
\usepackage{thm-restate}
\usepackage{xparse}

\usepackage[font=small,labelfont=bf]{caption}

\usepackage[dvipsnames]{xcolor}

\usepackage{babel}
\usepackage{csquotes}

\DeclareFontFamily{U}{mathx}{\hyphenchar\font45}
\DeclareFontShape{U}{mathx}{m}{n}{
      <5> <6> <7> <8> <9> <10>
      <10.95> <12> <14.4> <17.28> <20.74> <24.88>
      mathx10
      }{}
\DeclareSymbolFont{mathx}{U}{mathx}{m}{n}
\DeclareMathSymbol{\bigtimes}{1}{mathx}{"91}

\usepackage{tcolorbox}
\tcbuselibrary{skins,breakable}
\tcbset{enhanced jigsaw}

\definecolor{DarkRed}{rgb}{0.5,0.1,0.1}
\definecolor{DarkBlue}{rgb}{0.1,0.1,0.5}

\usepackage{nameref}
\definecolor{ForestGreen}{rgb}{0.1333,0.5451,0.1333}
\definecolor{Red}{rgb}{0.9,0,0}
\usepackage[linktocpage=true,
	colorlinks,
	linkcolor=DarkRed,citecolor=ForestGreen,
    bookmarks,bookmarksopen,bookmarksnumbered]
	{hyperref}
\usepackage[noabbrev,nameinlink,capitalize]{cleveref}
\crefname{property}{property}{Property}
\creflabelformat{property}{(#1)#2#3}
\crefname{equation}{Eq}{Eq}
\crefname{algocf}{Algorithm}{Algorithms}
\creflabelformat{equation}{(#1)#2#3}

\usepackage{bm}
\usepackage{url}
\usepackage{xspace}
\usepackage[mathscr]{euscript}
\usepackage{mathrsfs}

\usepackage{wrapfig}
\usepackage{tikz}
\usetikzlibrary{arrows}
\usetikzlibrary{arrows.meta}
\usetikzlibrary{shapes}
\usetikzlibrary{backgrounds}
\usetikzlibrary{positioning}
\usetikzlibrary{decorations.markings}
\usetikzlibrary{patterns}
\usetikzlibrary{calc}
\usetikzlibrary{fit}
\tikzset{vertex/.style={circle, black, fill=Yellow, line width=1pt, draw, minimum width=8pt, minimum height=8pt, inner sep=0pt}}

\usepackage[framemethod=TikZ]{mdframed}

\usepackage[margin=1in]{geometry}

\usepackage{soul}

\renewcommand{\paragraph}[1]{\medskip\noindent\textbf{#1}}

\newtheorem{theorem}{Theorem}
\newtheorem{lemma}{Lemma}[section]
\newtheorem{proposition}[lemma]{Proposition}
\newtheorem{corollary}[lemma]{Corollary}

\newtheorem{definition}[lemma]{Definition}
\newtheorem{problem}{Problem}

\newtheorem*{claim*}{Claim}
\newtheorem*{proposition*}{Proposition}
\newtheorem*{lemma*}{Lemma}
\newtheorem*{problem*}{Problem}

\crefname{lemma}{Lemma}{Lemmas}
\crefname{claim}{Claim}{Claims}
\crefname{enumi}{Step}{Steps}
\crefname{step}{Step}{Step}

\newtheorem{remark}[lemma]{Remark}

\theoremstyle{definition}

\newtheorem{mdproblem}{Problem}

\newtheorem*{mdproblem*}{Problem}
\newenvironment{Problem*}{\begin{mdframed}[hidealllines=false,innerleftmargin=10pt,backgroundcolor=gray!10,innertopmargin=5pt,innerbottommargin=5pt,roundcorner=10pt]\begin{mdproblem*}}{\end{mdproblem*}\end{mdframed}}
\newtheorem{mddefinition}[lemma]{Definition}

\renewcommand{\qed}{\nobreak \ifvmode \relax \else
      \ifdim\lastskip<1.5em \hskip-\lastskip
      \hskip1.5em plus0em minus0.5em \fi \nobreak
      \vrule height0.75em width0.5em depth0.25em\fi}

\renewcommand{\leq}{\leqslant}
\renewcommand{\geq}{\geqslant}
\renewcommand{\le}{\leqslant}
\renewcommand{\ge}{\geqslant}

\usepackage[linesnumbered,lined,algoruled,noend]{algorithm2e}\DontPrintSemicolon

\DeclarePairedDelimiter{\bracket}[]
\DeclarePairedDelimiter{\paren}()

\DeclarePairedDelimiter{\ceil}{\lceil}{\rceil}
\DeclarePairedDelimiter{\floor}{\lfloor}{\rfloor}

\DeclarePairedDelimiter{\set}{\{}{\}}

\DeclarePairedDelimiterXPP\lonenorm[1]{}\lVert\rVert{_1}{#1}

\DeclareMathOperator*{\distrib}{\mu}
\DeclareMathOperator*{\support}{supp}
\DeclareMathOperator*{\expect}{\mathbb{E}}
\DeclareMathOperator*{\variance}{Var}

\DeclarePairedDelimiterXPP{\TVD}[2]{\Delta_{\textnormal{\text{TVD}}}}(){}{#1, #2}
\DeclarePairedDelimiterXPP{\PTD}[2]{\Lambda}(){}{#1, #2}

\DeclarePairedDelimiterXPP{\Ot}[1]{\widetilde{O}}(){}{#1}
\DeclarePairedDelimiterXPP{\Omgt}[1]{\widetilde{\Omega}}(){}{#1}
\DeclarePairedDelimiterXPP{\BigO}[1]{O}(){}{#1}

\newcommand\given{ \nonscript\: \vert{} \nonscript\: \mathopen{} }

\NewDocumentCommand{\Prob}{sO{}E{_}{{}}m}{{\mathbb{P}}_{#3}
  \IfBooleanTF{#1}
  {\bracket*{#4}}
  {\bracket[#2]{#4}}
}

\NewDocumentCommand{\Exp}{sO{}E{_}{{}}m}{\expect_{#3}
  \IfBooleanTF{#1}
  {\bracket*{#4}}
  {\bracket[#2]{#4}}
}

\DeclarePairedDelimiterXPP{\Var}[1]{\variance}[]{}{#1}
\DeclarePairedDelimiterXPP{\Cov}[1]{\covariance}[]{}{#1}
\DeclarePairedDelimiterXPP{\eexp}[1]{\exp}(){}{#1}

\NewDocumentCommand{\Dist}{sO{}E{_}{{}}m}{\distrib_{#3}
  \IfBooleanTF{#1}
  {\paren*{#4}}
  {\paren[#2]{#4}}
}

\DeclarePairedDelimiterXPP{\Supp}[1]{\support}(){}{#1}

\DeclarePairedDelimiterXPP{\KL}[2]{\kldiv}(){}{#1 \;\delimsize\|\; #2}

\DeclarePairedDelimiterXPP{\Ent}[1]{\entropy}(){}{#1}
\DeclarePairedDelimiterXPP{\Inf}[2]{\inform}(){}{#1 \; ; \; #2}

\renewcommand{\epsilon}{\varepsilon}
\newcommand{\eps}{\varepsilon}

\newcommand{\poly}{\mbox{\rm poly}}

\renewcommand{\wp}{$\text{w.p.}$\xspace}
\newcommand{\whp}{$\text{w.h.p.}$\xspace}

\newenvironment{tbox}{\begin{tcolorbox}[
		enlarge top by=5pt,
		enlarge bottom by=5pt,
		 breakable,
		 boxsep=0pt,
                  left=4pt,
                  right=4pt,
                  top=10pt,
                  arc=0pt,
                  boxrule=1pt,toprule=1pt,
                  colback=white
                  ]}
{\end{tcolorbox}}

\SetKwProg{Fn}{function}{\string:}{}
\SetKwBlock{Repeat}{repeat}{}
\SetKwFor{ForWithoutDo}{for}{}{}

\SetKwFunction{NaiveAlgorithm}{NaiveAlgorithm}
\SetKwFunction{FreshColoring}{RefreshColoring}
\SetKwFunction{ComputeACD}{RefreshACD}
\SetKwFunction{UpdateACDInfo}{UpdateACDInfo}
\SetKwFunction{Fix}{Fix}
\SetKwFunction{Color}{Color}

\SetKwFunction{Insert}{Insert}
\SetKwFunction{Delete}{Delete}
\SetKwFunction{RecolorInsert}{RecolorInsert}
\SetKwFunction{RecolorDelete}{RecolorDelete}

\SetKwFunction{RecolorSparse}{RecolorSparse}
\SetKwFunction{RecolorMatching}{RecolorMatching}
\SetKwFunction{RecolorDense}{RecolorDense}
\SetKwFunction{AddAntiEdgeMatching}{IncrMatching}

\SetKwFunction{ColorDense}{ColorDense}
\SetKwFunction{OneShotColoring}{OneShotColoring}
\SetKwFunction{IncreaseMatching}{IncreaseMatching}
\SetKwFunction{FreshMatching}{FreshMatching}

\newcommand{\col}{\varphi}

\newcommand{\Vsparse}{V_{sp}}
\newcommand{\NS}{N_{S}}
\newcommand{\ND}{N_{V\setminus S}}
\newcommand{\LS}{L_{S}}
\newcommand{\hateps}{\eta} 
\title{Faster Dynamic $(\Delta+1)$-Coloring Against Adaptive Adversaries}
\author{Maxime Flin}
\address{Reykjavik University, Iceland.}
\email{maximef@ru.is}
\author{Magnús M. Halldórsson}
\address{Reykjavik University, Iceland.}
\email{mmh@ru.is}

\thanks{This work was supported by the Icelandic Research Fund grants 217965 and 2310015.}

\date{}

\begin{document}
\begin{abstract}
We consider the problem of maintaining a proper $(\Delta + 1)$-vertex coloring in a graph on $n$-vertices and maximum degree $\Delta$ undergoing edge insertions and deletions. We give a randomized algorithm with amortized update time $\Ot{ n^{2/3} }$ against adaptive adversaries, meaning that updates may depend on past decisions by the algorithm.
This improves on the very recent $\Ot{ n^{8/9} }$-update-time algorithm by Behnezhad, Rajaraman, and Wasim (SODA 2025) and 
matches a natural barrier for dynamic $(\Delta+1)$-coloring algorithms. 
The main improvements are in the densest regions of the graph, where we use structural hints from the study of distributed graph algorithms.

\end{abstract}
\maketitle

\section{Introduction}

In the $(\Delta+1)$-coloring problem, we are given a graph with maximum degree $\Delta$ and must assign to each vertex $v$ a color $\col(v) \in \set{1, 2, \ldots, \Delta+1}$ such that adjacent vertices receive different colors. Despite its simplicity in the classical RAM model, this problem is surprisingly challenging in more constrained models such as distributed graphs \cite{Linial92,BEPS16,HSS18,CLP20}, sublinear models \cite{ACK19,CFGUZ19,CDP21,AY25}, dynamic graphs \cite{BCHN18,HP22,BGKL22,BRW24} and even recently communication complexity \cite{FM24,CMNS24}.

This paper presents a \emph{fully dynamic} algorithm for the $(\Delta+1)$-coloring problem. 
The vertex set is fixed of size $n$, while edges can be inserted or deleted as long as the maximum degree remains at most $\Delta$ at all times.
We assume $\Delta$ is known in advance.
As the graph changes, our algorithm must maintain a $(\Delta+1)$-coloring. The time the algorithm takes to process an edge insertion or deletion is called the \emph{update time} and our aim is to minimize it.

The study of dynamic algorithms for $(\Delta+1)$-coloring was initiated by \cite{BCHN18}, which provided a randomized algorithm with $O(\log n)$ amortized update time. Independently, the authors of \cite{BGKL22} and \cite{HP22} provided randomized algorithms with $O(1)$ amortized update time. The caveat with those algorithms is that they assume updates are \emph{oblivious}: all edge insertions and deletions are fixed in advance and may not adapt to the colors chosen by the algorithm.

Maintaining a coloring against an adaptive adversary --- that may base the updates on past decisions by the algorithm --- is significantly harder, as the adversary can force the algorithm to recolor vertices at every update by connecting same-colored vertices.
Until very recently, no $o(n)$ update time algorithm was known for $(\Delta+1)$-coloring against adaptive adversaries.
In a recent breakthrough, Behnezhad, Rajaraman, and Wasim \cite{BRW24} presented a randomized algorithm with $\Ot{ n^{8/9} }$ amortized update time against adaptive adversaries\footnote{In this paper, we use $\Ot{f(n)} := O(f(n) \cdot \poly(\log n))$ to hide polylogarithmic factors.}. 
In this paper, we significantly improve the runtime, obtaining the following result.
\smallskip

\begin{theorem}
    \label{thm:main}
    There exists a randomized fully dynamic algorithm maintaining a $(\Delta+1)$-coloring against adaptive adversaries in $\Ot{ n^{2/3} }$ amortized update time with high probability.
\end{theorem}

At its core, our algorithm is based on a structural decomposition of the vertices into four layers so that vertices of the bottom layer can be colored efficiently by maintaining and searching small sets. Vertices of higher layers are colored by random color trials enhanced with a \emph{color stealing} mechanism. 
Specifically, when a vertex cannot use a color due to a lower-level neighbor, our algorithm steals the color and recolors the lower-level vertex instead.
Interestingly, the layers correspond to the order in which streaming and distributed algorithms build their colorings.

In addition to the runtime improvement, our algorithm approaches a natural barrier for dynamic coloring algorithms against adaptive adversaries. 
The algorithm of \cite{BRW24} was primarily limited by the challenge of efficiently recoloring dense regions of the graph.
In contrast, our algorithm shifts this bottleneck to the recoloring of sparser vertices and the handling of random color trials in general. We elaborate on this in the following discussion.

\subsection{Discussion \& Open Problems}

Before we explain why $n^{2/3}$ is a natural barrier, we first establish why $n^{1/2}$ poses a fundamental limit for dynamic algorithms based on random color trials. 
Observe that this limitation extends to simpler problems, such as $2\Delta$-coloring.

The bottleneck for recoloring a vertex lies in efficiently verifying whether a color is available. One possibility is to iterate through all neighbors in $O(\Delta)$ time; another is to iterate through color classes (i.e., the set of vertices with a given color). At best, all color classes contain $O(n/\Delta)$ vertices each. So, any algorithm that checks arbitrary or random colors using the former method when $\Delta$ is small and the latter when $\Delta$ is large therefore has $O(n^{1/2})$ update time at best. 
Recall that an adaptive adversary can trigger recoloring at every update, preventing any opportunity for amortization.


To reduce the number of colors to $\Delta+1$, the approach of \cite{EPS15,HSS18,ACK19,CLP20,BRW24} and ours is to compute a partial coloring in which (sparse) vertices have many pairs of neighbors colored the same. The randomized argument underlying this approach is fragile, requiring a complete rerun every $\Theta( \Delta )$ update. 
Since it is based on random color trials at every vertex, it runs in $O(n^2 / \Delta)$ time. Amortized and balanced with the $O(\Delta)$ update time for small $\Delta$, we obtain the $n^{2/3}$ update time.

\begin{problem}
Does there exist a fully dynamic $(\Delta+1)$-coloring with $o(n^{2/3})$ update time that remains robust against adaptive adversaries?
\end{problem}

\subsection{Organization of the Paper}
In \cref{sec:technical-overview}, we provide a high-level overview of our algorithm and compare it with \cite{BRW24}.
\Cref{sec:prelims} introduces the sparser-denser decomposition that forms the foundation of our approach.
The core of our dynamic algorithm, along with the proof of \cref{thm:main}, is presented in \cref{sec:dynamic}.
Finally, we defer the least novel part of our algorithm  -- the periodic recomputation of a fresh coloring of the whole graph -- to \cref{sec:fresh}.

\section{Overview}
\label{sec:technical-overview}

We present the key concepts underlying \cref{thm:main}, deliberately simplifying certain technical details for the sake of intuition. We focus on the case where $\Delta \geq n^{2/3}$, as an update time of $O(\Delta)$ is straightforward to achieve by examining the adjacency list of the vertices.

\paragraph{The Sparser-Denser Decomposition.}
We employ a modified version of the structural decomposition introduced by Reed \cite{Reed98}, and further explored in \cite{HSS18, ACK19}, which partitions the vertices into sparse vertices and dense clusters. We classify vertices into clusters based on their average density. Vertices in clusters with low average density are designated as \emph{sparser}, while those in clusters with higher average density are considered \emph{denser} (see \cref{def:refined-acd}). This method is based on structural observations about the mean local density within a cluster (see \cref{lem:dense-ext-sparsity}), and provides a significant efficiency gain compared to the specific form of decomposition used in \cite{BRW24}.

\paragraph{Coloring the Sparser Vertices.}
The strategy for coloring sparser vertices is essentially the same as in \cite{BRW24}. Updates are performed in \emph{phases} of length $t = \Theta(n^{2/3})$, with each phase beginning by computing a \emph{fresh coloring} of the graph. This ensures that the adversary cannot decrease the number of available colors below $t$ because each sparser vertex starts with $2t$ available colors (see \cref{prop:fresh-coloring}) and the algorithm recolors at most one vertex per update. As the adversary updates the graph, we use the observation from \cite{ACK19,BRW24} that a vertex with $t$ available colors can be recolored in $\Ot{ n^{2/3} }$ time by random color trials. Note that the cost of computing a fresh coloring, amortized over updates of the following phase, is $\Ot{ n^{2/3} }$.

\paragraph{Coloring the Denser Vertices.}
We treat two types of vertices differently in each dense cluster. We exploit the fact that most vertices have sparsity not much greater than the average sparsity of their cluster. These vertices --- the \emph{inliers} --- have $O(n^{2/3})$ neighbors outside their cluster and $O(n^{2/3})$ anti-neighbors within their cluster. This allows for a compact representation of the colors available to inliers, and a simple deterministic search suffices to recolor them. The remaining vertices --- the \emph{outliers} --- are recolored using random color trials, similar to the strategy for sparser vertices. The key idea is that outliers may \emph{steal} a color from an inlier neighbor, causing the inlier to be recolored as described above. Since most colors are available for stealing, only a few random trials are required. (See \cref{lem:recolor-denser} for details.)

\paragraph{Maintaining the Color Balance.}
To guarantee efficient color trials, we ensure that each color class contains $\Ot{ n^{1/3} }$ vertices. This condition holds for the sparser vertices by an adaptation of \cite{BRW24} (see \cref{lemma:dyn-sparse-balance}). For the denser vertices, we maintain a \emph{Dense Balance Invariant}: each cluster contains at most two vertices of each color. To ensure that enough colors are available for inliers, we maintain a second invariant.

The \emph{Matching Invariant} ensures that enough colors are repeated within each cluster. These repeated colors form a \emph{colorful matching} \cite{ACK19}, which can be maintained efficiently (again) using random color trials. With the Matching Invariant in place, we can ensure that inliers receive a color that is unique within their cluster, thereby preserving the Dense Balance Invariant.

It is noteworthy how extensively our algorithm employs \emph{color stealing}: sparser vertices can steal colors from matched vertices, which can steal from outliers, and outliers can steal from inliers.

\paragraph{Comparison with \cite{BRW24}.}
The parameter used in \cite{BRW24} to distinguish between sparser and denser vertices is the $\epsilon$ of the structural decomposition \cite{Reed98,HSS18}, ultimately set to $\eps = \frac{\Delta^{1/5}}{n^{2/5}}$.
While this choice allows for a uniform treatment of dense vertices, it is costly due to its quadratic dependence on (local) sparsity.
Specifically, sparse vertices have $\Omega(\epsilon^2 \Delta)$ available colors, which limits the length of each phase and increases the amortized cost of recomputing fresh colorings.
Additionally, maintaining the decomposition itself incurs an update time of $\Ot{\eps^{-4}}$.
In contrast, we keep $\eps$ \emph{constant} and instead use the average sparsity of dense clusters directly as a parameter to classify vertices as sparser or denser.
Our method employs a two-pronged approach to recoloring denser vertices (inliers vs.\ outliers), with both steps relatively simple.
While \cite{BRW24} reasons about dense vertices through perfect matchings in the palette graph of \cite{ACK19}, we work directly with vertices and colors.
The augmenting paths in their palette graph are similar to our color-stealing mechanism.

A notable technical difference concerns the colorful matching: while \cite{BRW24} maintains a \emph{maximal} matching, we find that a sufficiently large matching suffices, allowing for simpler updates.

Recent years have seen a number of exciting results on $(\Delta+1)$-coloring across models, all made possible by the power of random color trials and, in most cases, with the help of Reed's structural decomposition \cite{Reed98}. We find particularly interesting that analyzing the \emph{average sparsity/density} of its clusters (inspired by \cite{HNT21,FGHKN23}) enabled us to improve on \cite{BRW24} down to a natural barrier. 
In particular, the inliers/outliers dichotomy for balancing the cost of coloring dense and sparse vertices may be of independent interest. 
\subsection{Notation}
For an integer $k \geq 1$, we write $[k] = \set{1, 2, \ldots, k}$. Throughout the paper, the input graph has a fixed vertex set $V = [n]$. The neighborhood of $v$ is $N(v)$ and we write $\deg(v) = |N(v)|$ its degree. For a set $S \subseteq V$, let $\NS(v) = N(v) \cap S$ be the set of neighbors in the set.
We assume that $\Delta$ is known in advance and that, at all times, every vertex has degree at most $\Delta$.
A non-adjacent pair of nodes $u,v$ is said to be an \emph{anti-edge}.

A \emph{partial coloring} is a function $\col : V \to [\Delta + 1] \cup \set{ \bot }$, where $\col(v) = \bot$ indicates that $v$ is not colored. For a set $S$ of nodes, let $\col(S) = \{\col(v) : v \in S\}$.
A coloring is \emph{total} when all vertices are colored, i.e., when $\bot \notin \col(V)$. The coloring is \emph{proper} if, for every edge $\set{u,v}$, either $\col(u) \neq \col(v)$ or $\bot \in \set{\col(u), \col(v)}$. In this paper, all colorings are implicitly proper. The \emph{palette} of $v$ is $L(v) = [\Delta + 1] \setminus \col(N(v))$, i.e., the set of colors not used by neighbors of $v$. The colors of the palette are also said to be \emph{available} (for $v$).

We say that an event holds ``with high probability'' --- often abridged to ``\whp'' --- when it holds with probability at least $1 - n^{-c}$ for a sufficiently large constant $c > 0$.

The algorithm frequently uses classic sets and operations that can be implemented efficiently using, say, a red-black tree.
They support insertions, deletions, and uniform sampling in $O(\log n)$ worst-case update time.

\section{The Sparser-Denser Decomposition}
\label{sec:prelims}

\label{sec:prelims-acd}
Our algorithm uses a modified version of the sparse-dense decomposition introduced by \cite{Reed98} and further extended in \cite{HSS18,ACK19} for $(\Delta+1)$-coloring in distributed and streaming models. For some small $\eps\in(0,1)$, these decompositions partition the vertices into a set of $\Omega(\eps^2 \Delta)$-sparse vertices and a number of dense clusters called $\eps$-almost-cliques.
We introduce now the key definitions of sparsity and almost-cliques.
\begin{definition}
    \label{def:sparsity}
    Vertex $v$ is $\zeta$-sparse if $G[N(v)]$ contains at most $\binom{\Delta}{2} - \zeta\Delta$ edges.
\end{definition}
\begin{definition}
    A set $D \subseteq V$ of vertices is an $\eps$-almost-clique if
    \begin{enumerate}
        \item \label[part]{part:acd-size}
        $|D| \leq (1+\eps)\Delta$, and
        \item \label[part]{part:acd-neighbors}
        every $v\in D$ has $|N(v) \cap D| \geq (1 - \eps)\Delta$.
    \end{enumerate}
\end{definition}

For a vertex $v$ in an almost-clique $D$, the set of its neighbors outside $D$ and the set of its non-neighbors
inside $D$
play a key role in our algorithm. We denote
$E(v) = N(v) \setminus D$ and $A(v) = D \setminus (N(v) \cup \set{v})$ the sets of \emph{external-} and \emph{anti-neighbors}, respectively. We use $e_v = |E(v)|$ and $a_v = |A(v)|$ to denote their size.

As explained earlier, we differentiate between sparse and dense vertices differently from \cite{Reed98,HSS18,ACK19,BRW24}. Namely, it suffices for us that an almost-cliques is dense \emph{on average}.
To make this precise, let $e_D = \sum_{v\in D} e_v/|D|$ and $a_D = \sum_{v\in D} a_v/|D|$ be the average external- and anti-degree in $D$, respectively. We can now state the definition of our sparser-denser decomposition.

\begin{definition}
    \label{def:refined-acd}
    Let $\eps \in(0,1)$, $\zeta \in [1,\Delta]$ and $G$ a graph with maximum degree $\Delta$. An $(\eps,\zeta)$-\emph{sparser-denser decomposition} of $G$ is a partition of its vertices into sets $S, D_1, \ldots, D_r$ such that
    \begin{enumerate}
        \item \label[part]{part:refined-sparsity}
        every $v\in S$ is $\zeta$-sparse, and
        \item \label[part]{part:refined-denser}
        every $D = D_i$ is an $\eps$-almost-clique such that $a_D + e_D \leq O( \zeta / \eps^2 )$.
    \end{enumerate}
\end{definition}
To compare our decomposition to that of \cite{Reed98,HSS18,ACK19}, we use their core decomposition with a constant $\eps$ but choose a sparsity parameter $\zeta$ and let $S$ be the $\Omega(\eps^2\Delta)$-sparse vertices along with the almost-cliques for which $a_D + e_D \gg_\eps \zeta$. This results only in $\zeta$-sparse vertices due to a structural results on the sparsity of dense vertices (see \cref{lem:dense-ext-sparsity}).
Henceforth, we refer to vertices of $S$ as to the \emph{sparser} vertices and vertices of $D$ as the \emph{denser} vertices.

As explained in \cref{sec:technical-overview}, we wish to color most of the denser vertices by iterating over their external- and anti-neighbors. The vertices for which this is feasible are called inliers:
\begin{definition}
    \label{def:inliers}
    For each $D = D_i$ in an $(\eps,\zeta)$-sparser-denser decomposition, define
    \[
    I_D = \set{ v\in D : e_v \leq 8 e_D \text{ and } a_v \leq 8 a_D }
    \quad\text{and}\quad
    O_D = D \setminus I_D \ .
    \]
    where vertices of $I_D$ (resp.\ $O_D$) are called \emph{inliers} (resp.\ \emph{outliers}).
\end{definition}
Note that by Markov's inequality, at least three quarters of each $D_i$ are inliers.

\subsection{Maintaining the Decomposition}

In this section, we explain how the fully dynamic algorithm of \cite{BRW24} for maintaining Reed's decomposition against an adaptive adversary can be adapted to our needs:
\begin{lemma}
    \label{prop:refined-part}
    There exists a universal constant $\delta \in (0,1)$ for which the following holds. Let $\eps \in (0, 1/6)$ and $1 \leq \zeta \leq \eps^2 \cdot \delta \Delta$.
    There exists a randomized fully dynamic algorithm that maintains array $part[1 \ldots n]$ that induces a partition of the vertices into sets $S = \set{ v : part[v] = \bot}$ and $D_i = \set{v :  part[v] = i}$ such that, \whp,
    \begin{itemize}
    \item $(S, D_1, \ldots, D_r)$ is an $(\eps,\zeta)$-refined partition,
\item for each $D_i \neq \emptyset$, it maintains the set $F_{D_i}$ of anti-edges in $G[D_i]$, and
    \item for each $v\in D_i$, it maintains the sets $A(v)$ and $E(v)$.
    \end{itemize}
    The amortized update time against an adaptive adversary is $O(\eps^{-4}\log n)$ with high probability.
\end{lemma}

Our algorithm builds on an $\eps$-almost-clique decomposition, as defined by \cite{ACK19}: a vertex partition $\Vsparse, C_1, C_2, \ldots$
where the nodes in $\Vsparse$ have $\Omega(\eps^2 \Delta)$-sparsity and the $C_i$'s are $\Theta( \eps )$-almost-cliques. We filter out the $C_i$'s with sparsity more than $\zeta$ and add them into $\Vsparse$ to obtain the set $S$.
The key structural result for our partition is the following lemma:

\begin{lemma}
    \label{lem:dense-ext-sparsity}
    Let $\eps \in (0,1/75)$ and $\delta \in (0,1)$.
    Suppose $\Vsparse, C_1, C_2, \ldots, C_r$ is a vertex partition such that every vertex in $\Vsparse$ is $( \eps^2 \cdot \delta\Delta )$-sparse and each $C_i$ is an $\eps$-almost-clique. Then, every $v\in C_i$ such that $e_{C_i} + a_{C_i} \geq 16/(\delta\eps^2)$ is $\frac{\delta \eps^2}{50} ( a_{C_i} + e_{C_i} )$-sparse.
\end{lemma}

\begin{proof}
    Let $v\in C = C_i$ and $s := e_C + a_C$ such that $s \geq 16/(\delta\eps^2)$.
    First, we argue that every $v$ with $e_v \geq 2/(\delta\eps^2)$ is $\max\set{\delta \eps^2/4 \cdot e_v, (1-3\eps)/2 \cdot a_v}$-sparse. In particular, if $e_v \geq s/8$ or $a_v \geq s/24$, then $v$ is at least $(\frac{ \delta \eps^2  }{ 50 } \cdot s)$-sparse.
    
    It was already proven in \cite[Lemma 6.2]{HKMT21} that $v$ is at least $(1 - 3\eps)/2 \cdot a_v$-sparse. Our definition of external degree, counting both sparse and dense neighbors outside $C$, differs from \cite{HKMT21} so we\footnote{A proof of this fact has already appeared in a technical report by Halld\'orsson, Nolin and Tonoyan \cite[Lemma 3]{HNT21}. Since it is a short and important proof, we include it here for completeness.} prove that $v$ is at least $\delta\eps^2/4 \cdot e_v$-sparse. We may assume that $G[N(v)]$ contains at least $\binom{\Delta}{2} - \delta\eps^2/2 \cdot e_v \Delta$ edges, for the claim otherwise trivially holds. By handshaking lemma, the number of edges in $G[N(v)]$ is
    \begin{equation}
        \label{eq:upperbound-GN-v}
    \frac{1}{2}\sum_{u\in N(v)} |N(u) \cap N(v)|
    \leq \frac{1}{2}\sum_{u\in N(v)} \Delta - |N(v) \setminus N(u)|
    \leq \binom{\Delta}{2} + \frac{\Delta}{2} - \sum_{u\in N(u)} |N(v) \setminus N(u)| \ .
    \end{equation}
    So we lower bound $|N(v) \setminus N(u)|$ for each $u \in E(v)$. If $u \in C_j \neq C$, then it has at most $\eps\Delta$ shared neighbors with $v$, so $|N(v) \setminus N(u)| \geq (1-\eps)\Delta$. If $u\in \Vsparse$, then $G[N(u)]$ contains at most $\binom{\Delta}{2} - \delta\eps^2\Delta$ many edges by definition. On the other hand, of the edges in $G[N(v)]$, at most $|N(v) \setminus N(u)|\Delta$ are not in $G[N(u)]$. Using our assumption on the number of edges in $G[N(v)]$, this means that $\binom{\Delta}{2} - \delta\eps^2/2 \cdot e_v\Delta \leq \binom{\Delta}{2} - \delta\eps^2 \Delta^2 + |N(v) \setminus N(u)| \Delta$, so that $|N(v) \setminus N(u)| \geq \delta\eps^2/2 \cdot \Delta$. Since all external neighbors contribute at least $\delta\eps^2/2\cdot \Delta$ to the sum in \cref{eq:upperbound-GN-v}, the number of edges in $G[N(v)]$ is at most $\binom{\Delta}{2} - (\delta\eps^2/2 \cdot e_v - 1/2)\Delta$. If $e_v \geq 2/(\delta\eps^2)$, this implies that $v$ is $(\delta\eps^2/4 \cdot e_v)$-sparse.

    We henceforth assume that $e_v \leq s/8$ and $a_v \leq s/24$. Now, there are two cases, depending on which of $e_C$ or $a_C$ dominates.

    If $e_C \geq 2a_C$, then $e_v \leq s/8 \leq e_C/4$. Since $\deg(v) +1 = |C| + e_v - a_v$ holds for all vertices in $C$, it also holds on average and we get that $\Delta + 1 \geq |C| + e_C - a_C \geq |C| + e_C/2$. Using $\deg(v)+1 = |C| + e_v - a_v$ again, we get $\deg(v) \leq |C| - 1 + e_v \leq \Delta - (e_C/2 - e_v) \leq \Delta - e_C/4$. Since a vertex with degree at most $\Delta - x$ is $(x/2)$-sparse, $v$ is $e_C/8$-sparse, which is $s/16$-sparse.

    If $e_C \leq 2a_C$, then $a_v \leq s/24 \leq a_C/8$. Observe that $C$ contains $a_C|C|/2$ anti-edges, of which at most $a_v \cdot 2\eps\Delta$ may have an endpoint out of $N(v) \cap C$. So $G[N(v) \cap C]$ contains at least $a_C|C|/2 - 2 a_v \eps\Delta \geq a_C\Delta/8$ anti-edges for $\eps < 1/2$. So $v$ is $a_C/8$-sparse which is $s/24$-sparse.
\end{proof}

To maintain our sparse-denser decomposition, we use the algorithm of \cite{BRW24} for the almost-clique decomposition.

\begin{proposition}[{\cite[Theorem 4.1]{BRW24}}]
    \label{prop:acd}
    For every $n \gg 1$ and $\hateps \in (0, 3/50)$, there exists a fully dynamic randomized algorithm that maintains a vertex partition $\Vsparse, C_1, \ldots, C_r$ of any $n$-vertex graph with maximum degree $\Delta$ such that, \whp,
    \begin{enumerate}
        \item vertices of $\Vsparse$ are $\Omega( \hateps^2\Delta )$-sparse,
        \item each $C_i$ is a $(10 \hateps)$-almost-clique,
        \item for every vertex it maintains the sets $N(v) \cap \Vsparse$ and $N(v) \cap (C_1 \cup \ldots \cup C_r)$, and
        \item for every almost-clique $C_i$, it maintains the set $F_i$ of anti-edges in $C_i$.
    \end{enumerate}
   The amortized update time against an adaptive adversary is $O(\hateps^{-4}\log n)$ with high probability.
\end{proposition}

We can now prove that a sparser-denser decomposition can indeed be maintained.

\begin{proof}[Proof of \cref{prop:refined-part}]
    Define $\hateps = \eps/10$ and let $\delta$ be the universal constant such that vertices of $\Vsparse$ are $(\eps^2 \cdot \delta\Delta)$-sparse in \cref{prop:acd}.
    We use the algorithm of \cref{prop:acd} with $\hateps = \eps/10$ and define $S$ as the set containing all vertices of $\Vsparse$ and the vertices from almost-cliques $C = C_i$ with $a_C + e_C \geq \frac{50}{\delta \eps^2} \zeta$.
    Since we assume in \cref{prop:refined-part} that $\zeta \leq \eps^2 \cdot \delta\Delta$, every vertex of $\Vsparse$ is $\zeta$-sparse. Every vertex of $S \setminus \Vsparse$ is $\zeta$-sparse by \cref{lem:dense-ext-sparsity}. The remaining $\eps = (10\hateps)$-almost-cliques are renamed $D_i$ and verify that $a_{D_i} + e_{D_i} \leq \frac{50}{\delta \eps^2} \zeta = O(\zeta/\eps^2)$ by construction. So the resulting decomposition is indeed an $(\eps,\zeta)$-sparser-denser decomposition.

    The algorithm of \cref{prop:acd} already maintains the set of anti-edges in each $D_i$, so it only remains to explain how we maintain sets of external- and anti-neighbors for each denser vertex. Every time the algorithm of \cref{prop:acd} modifies an element in $N(v) \cap \Vsparse$ or $N(v) \cap (C_1 \cup \ldots \cup C_r)$, we inspect the change and modify $E(v)$ accordingly. Similarly, when the algorithm modifies a set of anti-edges $F_i$ by inserting or deleting some anti-edges $\set{u,v}$, we modify the sets $A(v)$ and $A(u)$ accordingly. Since each inspection (and possible modification) takes $O(\log n)$ time, it does not affect the complexity of the algorithm by more than a constant factor.
\end{proof}

\subsection{Clique Palette, Colorful Matching \& Accounting Lemma}
\label{sec:colorful-matching}

Similarly to \cite{ACK19,FGHKN24}, we wish to assign dense vertices colors that are unused by vertices of their almost-clique. Together these form the \emph{clique palette}.
\begin{definition}
    \label{def:clique-palette}
    For any (possibly partial) coloring $\col$ of $G$, and almost-clique $D$, the \emph{clique palette} $L(D)$ is the set of colors unused by vertices of $D$. More formally, $L(D) = [\Delta + 1] \setminus \col(D)$.
\end{definition}

To ensure the clique palette always contains colors, we give the same color to a pair of vertices in the almost-clique.
We say a color is \emph{repeated} or \emph{redundant} in $D \subseteq V$ when at least two vertices of $D$ hold it. Assadi, Chen, and Khanna \cite{ACK19} introduced the notion of a \emph{colorful matching}: if $M$ colors are repeated in $D$, there exists an $M$-sized anti-matching in $G[D]$ where the endpoints of a matched anti-edge have the same color.

Each redundant color in $D$ gives its vertices some slack, while using the clique palette forbids at most one color for each anti-neighbor. Accounting for both, we obtain the following bound on the number of colors in the clique palette available for each vertex:
\begin{lemma}[Accounting Lemma]
    \label{prop:accounting}
    Let $\col$ be a (possibly partial) coloring and $D$ an almost-clique with a colorful matching of size $M$. For every $v\in D$,
    the number of colors available for $v$ in the clique palette is at least
    \[
        |L(D) \cap L(v)|
        \geq \Delta - \deg(v) + M - a_v + [\#\text{uncolored vertices in $D \cup E(v)$}] \ .
    \]
\end{lemma}

\begin{proof}
    For succinctness, let $k$ denote the number of uncolored vertices in $D\cup E(v)$. Observe how the number of vertices in $D \cup E(v)$ relates to the size of the colorful matching by partitioning vertices according to their colors:
    \begin{equation}
        \label{eq:accounting}
        |D \cup E(v)|
        = k + \sum_{\chi \in [\Delta+1]} |\set{ v \in D \cup E(v): \col(v) = \chi}|
        \geq
        k
        + |\col(D \cup E(v))| + M \ ,
    \end{equation}
    where the last inequality comes from the fact that at least $M$ colors are used by at least two vertices in $D$.
    Comparing the number of colors in the clique palette available to $v$ to the size of $D$ shows the contribution of redundant colors and uncolored vertices:
    \begin{align*}
        |L(D) \cap L(v)|
        = \Delta + 1 - |\col(D \cup E(v))|
        &= \Delta + 1 - (|D| + e_v) + |D \cup E(v)| - |\col(D \cup E(v))| \\
        &\geq \Delta + 1 - |D| - e_v + k + M \ . \tag{by \cref{eq:accounting}}
    \end{align*}
    To complete the proof, observe that $\deg(v) + 1 = |D| + e_v - a_v$, so that $\Delta + 1 - |D| - e_v$ simplifies to $(\Delta - \deg(v)) - a_v$. Hence, the lemma.
\end{proof}
 \section{The Dynamic Algorithm}
\label{sec:dynamic}

In this section, we prove \cref{thm:main}. We begin with the definition of phases and describe the properties of the coloring at the beginning of each phase. We then explain how the algorithm maintains a proper coloring over one phase. The analysis has three parts: first, we show that the algorithm maintains some invariants 
throughout the phase
(\cref{sec:invariants}); second, we prove that recoloring denser vertices is efficient (\cref{sec:denser}); finally, we analyze the re-coloring of sparser vertices (\cref{sec:sparser}). Together, they imply \cref{thm:main} (\cref{sec:proof-theorem}).

When $\Delta \leq O(n^{2/3})$, the naive deterministic algorithm that scans the entire neighborhood of vertices before recoloring them (when necessary) has worst-case update time $O(\Delta) \leq O(n^{2/3})$. We henceforth assume that $\Delta \geq \Omega( \zeta/\eps^2 ) = \Omega(n^{2/3})$, for a sufficiently large constant where $\eps := 1/110$ and $\zeta := n^{2/3}$ are the parameters of the sparser-denser decomposition.

\paragraph{Phases \& Fresh Coloring.}
We divide updates into \emph{phases} of $t$ insertions and deletions and analyse the algorithm phase by phase.
At the beginning of each phase, we compute a fresh coloring of the graph.
\begin{quote}
    A phase is $t := \gamma \cdot \zeta$ updates long, where $\gamma$ is the constant from \cref{prop:fresh-coloring}.
\end{quote}
The algorithm recolors at most one vertex in $S$ per update, so the adversary can only remove one color per update from $\LS(v) := [\Delta+1] \setminus \col[\NS(v)]$.
By recomputing a fresh coloring every $t$ updates, we ensure that sparser vertices always have $t = \gamma \cdot \zeta$ colors available,  as they start the phase with more than $3t$ available colors.

\begin{restatable}{proposition}{PropFreshColoring}
    \label{prop:fresh-coloring}
    There exists a universal constant $\gamma \in (0,1)$ such that the following holds.
    For $\zeta \geq \Omega( \log n )$ for some large hidden constant, let $S, D_1, \ldots, D_r$ be a $(\eps, \zeta)$-decomposition of an $n$-vertex graph with maximum degree $\Delta$.
    W.h.p., \FreshColoring runs in $\Ot{ n + n^2 / \zeta }$ time, producing a total coloring $\col$ such that
    \begin{itemize}
        \item (Sparse Slack) every sparse vertex $v$ has $|\LS(v)| \geq 3 \gamma \cdot \zeta$,
\item (Sparse Balance) $|\Phi[\chi] \cap S| = O( n / \zeta + \log n)$ for every $\chi \in [\Delta + 1]$,
        \item (Dense Balance) in each $D_i$, every color is used at most twice, and
        \item (Matching) each $D_i$ has at least $\floor*{ 8 a_{D_i} }$ redundant colors.
    \end{itemize}
\end{restatable}
The description of \FreshColoring and the proof of \cref{prop:fresh-coloring} are deferred to \cref{sec:fresh} to preserve the flow of the paper.
Henceforth, we say that ``\FreshColoring succeeded'' if it ends and the coloring it produces has all the aforementioned properties.
Throughout each phase, we maintain the Dense Balance and Matching properties given by \FreshColoring as invariants --- and refer to them as the \emph{Dense Balance Invariant} and \emph{Matching Invariant}.

During a phase, it helps to keep the sparser-denser partition fixed. A simple adaptation of the algorithm from \cref{prop:refined-part} allows us to make this assumption. We emphasize that during a phase, the adversary may increase or decrease external- and anti-degrees of vertices and their averages. However, given the phase length, it does not materially affect the decomposition.

\begin{lemma}
    \label{lem:dyn-refined-part}
    Let $\delta$, $\eps$ and $\zeta$ be as in \cref{prop:refined-part}.
    There is an algorithm with $O(\eps^{-4} \log n)$ amortized update time that maintains a vertex partition such that
    \begin{itemize}
        \item at the beginning each phase, the partition is an $(\eps,\zeta)$-decomposition,
        \item the vertex partition does not change during phases,
        \item at all times, denser vertices have $e_v \leq 2\eps\Delta$ and $a_v \leq 3\eps\Delta$, and
        \item for each $D = D_i$, the values of $a_D$ and $e_D$ increase or decrease by at most one within a phase.
\end{itemize}
\end{lemma}

\begin{proof}
    Use \cref{prop:refined-part} to maintain a sparser-denser decomposition with $O(\eps^{-4}\log n)$ update time except that, during phases, instead of processing updates, we store them in a stash. When a new phase begins, we process all updates stored in the stash and obtain an $(\eps,\zeta)$-refined partition of the graph at that time. Clearly the vertex partition is fixed during phases. Now, observe that at the beginning of a phase, since every $D = D_i$ is an $\eps$-almost-clique, the denser vertices have $e_v \leq \eps\Delta$ and $a_v \leq 2\eps\Delta$. Since the adversary can insert/delete at most $t \leq \zeta \leq \eps\Delta$ edges per phase, we always have that $e_v \leq \eps\Delta + t \leq 2\eps\Delta$ and $a_v \leq 2\eps\Delta + t \leq 3\eps\Delta$. Similarly, the adversary may increase/decrease the sums of external- and anti-degrees by at most $t \leq \eps\Delta$, hence the values of $a_D$ and $e_D$ increase/decrease by up to an additive $t/|D| \leq 2\eps$ term over the course of each phase.
\end{proof}

\paragraph{The Dynamic Algorithm.}
Updates are triggered by calls to \Insert and \Delete. They maintain the sparser-denser partition as described in \cref{lem:dyn-refined-part} and, when a new phase begins, recompute a coloring from scratch using \FreshColoring. During phases, \Insert and \Delete update the data-structures and respectively call \RecolorInsert and \RecolorDelete to handle color conflicts.

Upon insertion of $\set{u,v}$, \RecolorInsert verifies if it creates a color conflict (line \ref{line:recolor-insert-test-conflict}), in which case it uncolors $u$. In most cases, it suffices to recolor $u$ using \RecolorSparse or \RecolorDense depending on whether $v\in S$ or not (lines \ref{line:begin-recolor-endpoint} to \ref{line:end-recolor-endpoint}). The one exception to this concerns \emph{matched vertices}. A denser vertex is \emph{matched} if and only if its color is repeated by another vertex of its almost-clique. Maintaining sufficiently many matched vertices in each almost-clique (see Matching Invariant) is necessary for \RecolorDense. If, after an update, there are not enough matched vertices --- either because the almost-clique lost an anti-edge (line \ref{line:check-kill-anti-edge}) or because $a_D$ increased (line \ref{line:delete-if-matching-inv}) --- we call \AddAntiEdgeMatching{D}, which creates a new redundant color in $D$. To same-color pairs of (non-adjacent) vertices in $D$, we use \RecolorMatching. It proceeds similarly in recoloring sparse vertices and outliers: random color trials in $[\Delta+1]$ and possibly stealing a color from an inlier in $D$.

\begin{algorithm}
    \caption{Updates handlers during a phase, when $\Delta \geq \Omega( n^{2/3} )$}
    \label{alg:del}
    \label{alg:insertion}
    \Fn{\RecolorDelete{u,v}}{
        \lIf*{$part[u] = part[v] = i$ and $|M_{D_i}| < \floor{ 8 a_{D_i} }$} {\AddAntiEdgeMatching{$D_i$}
        \label{line:delete-if-matching-inv}}
    }

    \Fn{\RecolorInsert{u,v}}{

        \If{$\col[u] = \col[v]$ \label{line:recolor-insert-test-conflict}}{
            \Color{u, $\bot$} \tcp*[f]{uncolor $u$}\;
            \If{$matched[u] = w \neq \bot$}{\label{line:check-kill-anti-edge}
                $matched[u], matched[w] \gets \bot$
                    \tcp*[f]{unmatch vertices}\;
                $M_D \gets M_D - \col[w]$\tcp*[f]{remove the color from the matching}\;
            }
        }

        \lIf{$part[u] = part[v] = i$ and $|M_{D_i}| < \floor{ 8 a_{D_i} }$} {\AddAntiEdgeMatching{$D_i$} \label{line:insert-if-matching-inv}}

        \If{$\col[u] = \bot$}{
            \label{line:begin-recolor-endpoint}
            \lIf{$u \in S$}{
                \RecolorSparse{u}
            }\lElseIf{$matched[u] = w \neq \bot$}{
                \RecolorMatching{u, w}
                \label{line:insert-calls-recolor-matching}
            }
            \lElse*{
                \RecolorDense{u}
                \label{line:end-recolor-endpoint}
            }
        }
    }
\end{algorithm}

Let us list the data-structures we use. The coloring is stored as an array $\col[1\ldots n]$ such that $\col[v]$ is the color of $v$. We also maintain the color class $\Phi[\chi] = \set{ v\in V: \col[v] = \chi}$ for each $\chi \in [\Delta + 1]$. When we change the color of a vertex, we use the function \Color{v, $\chi$} that sets $\col[v]$ to $\chi$ and updates color classes accordingly.
To keep track of redundant colors, we maintain:
\begin{itemize}
    \item an array $matched[1 \ldots n]$ such that $matched[u] = v \neq \bot$ if $u$ and $v$ are in the same almost-clique and colored the same, otherwise $matched[u] = \bot$;
    \item for each almost-clique $D_i$, the set $M_{D_i}$ of repeated colors in $D_i$.
\end{itemize}

We now explain how the recoloring procedures used in \cref{alg:insertion} work for sparser, denser and matched vertices.

\paragraph{Sparser Vertices.}
For recoloring a vertex in $S$, we sample random colors in $[\Delta+1]$ and recolor the vertex with the first color that is not used by neighbors of $S$ and used by at most one denser neighbor. As we shall see, there are at least $\gamma \cdot \zeta$ such colors with high probability\footnote{For convenience, if no available color exist, the algorithm loops forever. As we explain later, this occurs with probability $1/\poly(n)$. If one wishes the algorithm to have finite expected update time, the algorithm can be modified to restart a phase if no available color was found after $\Ot{ n^{1/3} }$ samples.}, so $\Ot{ \Delta /\zeta }$ color tries suffice with high probability. To verify if a color is available, the algorithm loops through its color class (line \ref{line:recol-sparse-test}). Finally, if the color was used by a \emph{unique} denser neighbor $w$ (line \ref{line:recol-sparse-steal}), we steal the color from $w$ and recolor the denser vertex using \RecolorDense.

\begin{algorithm}
    \caption{For a sparse vertex $v$}
    \label{alg:recolor-sparse}
    \Fn(\tcp*[f]{$\col[v] = \bot$ and $v\in S$}){\RecolorSparse{v}}{
        Sample $\chi \in [\Delta + 1]$ and $w \gets \bot$\label[line]{line:recol-sparse-sample}\;
        \For{$u\in \Phi[\chi]$ \label{line:recolor-sparse-iterate-color-class}}{
            \lIf{$u \in \NS(v)$ or ($u\in \ND(v)$ and $w \neq \bot$)}{go back to line \ref{line:recol-sparse-sample} \label{line:recol-sparse-test}}
            \lElseIf{$u\in \ND(v)$ and $w = \bot$}{$w \gets u$}
        }
        \If{$w \neq \bot$\label{line:recol-sparse-steal}}{
            \Color{w, $\bot$},
            \Color{v, $\chi$} \tcp*[f]{steal from $w$}\;
            \lIf{$matched[w] = \bot$}{ \RecolorDense{w} }
            \lElse{ \RecolorMatching{w, matched[w]} }
        }
        \lElse{
            \Color{v, $\chi$}
        }
    }
\end{algorithm}

\paragraph{Denser Vertices.}
We now explain how we recolor denser \emph{unmatched} vertices.
Recall that inliers are such that $a_v \leq 8a_C$ and $e_v \leq 8e_C$ (\cref{def:inliers}). Note that the algorithm does not maintain those sets explicitly, as it can simply test for those inequalities. Importantly, inliers always have colors available in the clique palette by the Accounting Lemma (\cref{prop:accounting}). We therefore color them by computing the set of colors used by external neighbors and searching for an available color in $L(D)$ (line \ref{line:recol-denser-inlier}).

We color outliers like the sparser vertices (lines \ref{line:recolor-denser-outliers-start} to \ref{line:recolor-denser-outliers-end}): by sampling colors in $[\Delta+1]$. We allow outliers to steal the color of an \emph{unmatched inlier} and then recolor the inlier instead (line \ref{line:recolor-denser-call-inlier}).

\begin{algorithm}
    \caption{For a denser $v \in D$ (which needs recoloring) and is unmatched}
    \label{alg:recolor-denser}
    \Fn(\tcp*[f]{$\col[v] = \bot$ and $matched[v] = \bot$}){\RecolorDense{v}}{
\If(\tcp*[f]{color an outlier}){$v \notin I_D$ \label{line:recolor-denser-outliers-start}}{
            Sample $\chi \in [\Delta + 1]$ \label[line]{line:recolor-denser-outlier-sample}\;
            \lIf{$\chi \in M_D$}{go to line \ref{line:recolor-denser-outlier-sample} \label{line:recolor-denser-if-free-outlier-1}}
            \lForWithoutDo*{$u\in \Phi[\chi]$\label{line:recolor-denser-color-class}}{
                \lIf{$u \in N(v)$ but $u \notin I_D$}{
                    go to line \ref{line:recolor-denser-outlier-sample}
                }
                \label{line:recolor-denser-if-free-outlier-2}
            }
            \lIf{$\Phi_D[\chi] \cap I_D = \set{u}$}{
                \Color{u, $\bot$}, \Color{v, $\chi$},
                \RecolorDense{u}
                \label{line:recolor-denser-call-inlier}
            }\lElse*{\Color{v, $\chi$}}
            \label{line:recolor-denser-outliers-end}
        }
        \lElse*{
            \Color{v, $\chi$} with any $\chi \in L(D) \setminus \col[E(v)]$
            \tcp*[f]{$v\in I_D$}
            \label{line:recol-denser-inlier}
        }
    }
\end{algorithm}

To implement \RecolorDense efficiently, for each almost-clique $D$, we maintain
\begin{itemize}
\item the set of unused colors $L(D) = [\Delta + 1] \setminus \col[D]$, a.k.a.\ the clique palette,
    \item the set of vertices $\Phi_D[\chi]$ in $D$ with color $\chi$.
\end{itemize}
Those data-structures are easy to maintain with $O(\log n)$ worst-case update time because each update affects at most two vertices --- thus at most two almost-cliques --- and each set changes by at most one element.

\paragraph{Recoloring a Matched Anti-Edge.}
For a pair $u,v \in D$, \RecolorMatching{u,v} colors --- or recolors --- $u$ and $v$ the same. As we cannot ensure they share many available colors in the clique palette, we allow $u$ and $v$ to pick any color that is used by neither a matched vertex (line \ref{line:recolor-matching-test-M}) nor external neighbors (line \ref{line:recolor-matching-test-ext}). Since we only need few repeated colors, they always have $\Omega( \Delta )$ colors to choose from and we can find one using $O(\log n)$ random samples. However, when doing so, we might have to --- in fact, most likely will --- recolor an unmatched vertex (line \ref{line:recolor-matching-uncolor}).

\RecolorMatching is used whenever matched vertices need to get recolored (line \ref{line:insert-calls-recolor-matching} in \cref{alg:insertion}) or when the matching size needs to be increased (line \ref{line:maintain-recolor-matching} in \cref{alg:recolor-matching}). To increase the size of the matching, \AddAntiEdgeMatching needs to find a pairs $\set{u,v}$ of unmatched anti-neighbors to same-color. Since we maintain a matching of size $\Theta(a_D)$ while $D$ contains $\Theta(a_D\Delta)$ and anti-degrees are at most $O( \eps\Delta )$, a random anti-edge in $F_D$ has both endpoints unmatched with constant probability.

\begin{algorithm}
    \caption{Same-colors a pair $u,v\in D$}
    \label{alg:recolor-matching}
    \Fn{\RecolorMatching{u, v}}{
        \lIf{$matched[u] = matched[v] = \bot$}{
            $matched[u] \gets v$, $matched[v] \gets u$
        }
        Sample $\chi \in [\Delta + 1]$ \label{line:recolor-matching-sample}\;
        \lIf{$\chi \in M_D$ \label{line:recolor-matching-test-M}}{go to line \ref{line:recolor-matching-sample}}
        \lForWithoutDo*{$w \in \Phi[\chi]$ \label{line:recolor-matching-test-ext}}{
            \lIf{$w \in E(v) \cup E(u)$}{go to line \ref{line:recolor-matching-sample}}
        }
        $M_D \gets M_D + \chi$\;
        \If(\tcp*[f]{steal the color of an unmatched vertex}){$\Phi_D[\chi] = \set{w}$}{\Color{w, $\bot$}, \Color{u, $\chi$}, \Color{v, $\chi$} and \RecolorDense{w} \label{line:recolor-matching-uncolor}}
        \lElse{\Color{u, $\chi$} and \Color{v, $\chi$}}
    }

    \Fn{\AddAntiEdgeMatching{D}}{
        Sample $\set{u,v} \in F_D$ uniformly at random \label{line:resample-anti-edge}\;
        \lIf{$matched[u] = matched[v] = \bot$}{
            \RecolorMatching{u, v} \label{line:maintain-recolor-matching}
        }\lElse{go to line \ref{line:resample-anti-edge}}
    }
\end{algorithm}

\subsection{Dense Invariants}
\label{sec:invariants}

We begin by showing that our recoloring algorithm maintains the dense invariants (\cref{lem:recolor-denser-matching-invariants}). More interestingly, we show in \cref{lem:matching-inv} that when the Matching Invariant is broken by an update, one call to \AddAntiEdgeMatching suffices to fix it.

\begin{lemma}
    \label{lem:recolor-denser-matching-invariants}
    Suppose the invariants hold before a call to \RecolorDense{v} with $matched[v]=\bot$ or \RecolorMatching{u,v} or \RecolorSparse{v}. If the call ends, it maintains the invariants.
\end{lemma}
\begin{proof}
    First, consider a call to \RecolorDense{v}. If $v \in I_D$, since $v$ is unmatched, recoloring $v$ does not affect the Matching Invariant; it maintains the Dense Balance Invariant because $v$ gets a color from $L(D)$, i.e., a color that is unused in $D$. If $v \notin I_D$, again since $v$ is unmatched it maintains the Matching Invariant. \RecolorDense{v} colors $v$ with a color which is either unused in $D$ or used by an unmatched inlier (lines \ref{line:recolor-denser-if-free-outlier-1}-\ref{line:recolor-denser-if-free-outlier-2}). If the color is free, the Dense Balance is maintained; if the color is used by an unmatched inlier $w$, we uncolor $w$ and call \RecolorDense{w} (line \ref{line:recolor-denser-call-inlier}). As we explained before, recoloring an unmatched inlier preserves invariants.

    Consider now a call to \RecolorMatching{u,v}. Observe that it can only increase the number of matched vertices, so it maintains the Matching Invariant. It colors $u$ and $v$ with some $\chi \notin M_D$. So $\chi$ is either free or used by a unique vertex in $D$. If a vertex $w$ uses $\chi$, the algorithm uncolors $w$ and calls \RecolorDense{w} (line \ref{line:recolor-denser-call-inlier}), which maintains the invariants as we previously explained.

    Finally, a call to \RecolorSparse{v} where $v\in S$ recolors $v$, which does not affect the invariants. It may steal a color to one denser neighbor $w$. If $w$ is unmatched \RecolorDense{w} maintains the invariant. If $w$ is matched, the algorithm calls \RecolorMatching{w, matched[w]} which can only end if $w$ and its matched vertex are same-colored. So the call to \RecolorSparse maintains the invariants.
\end{proof}

\cref{lem:recolor-denser-matching-invariants} shows that \RecolorInsert and \RecolorDelete maintain the Dense Balance Invariant since it is only affected when vertices are recolored. On the other hand, the Matching Invariant depends on the values of $|M_D|$ and $a_D$, which are affected by insertions or deletions of edges.
\begin{lemma}
    \label{lem:matching-inv}
    Consider a fixed update \RecolorInsert{u,v} or \RecolorDelete{u,v} before which the invariant holds. When \RecolorDelete{u,v} ends or when \RecolorInsert{u,v} reaches line \ref{line:begin-recolor-endpoint}, the invariants holds.
\end{lemma}

\begin{proof}
    Assume that $|M_D| < \floor{ 8 \cdot a_D }$ when the algorithm reaches line \ref{line:delete-if-matching-inv} of \RecolorDelete or line \ref{line:insert-if-matching-inv} of \RecolorInsert. It must be that $|M_D| = \floor{ 8 \cdot a_D} - 1$ because the invariant held before the update and
    \begin{itemize}
        \item inserting an edge cannot increase $a_D$ and may decrease $M_D$ by one if the edge inserted was between matched vertices (line \ref{line:check-kill-anti-edge});
        \item deleting an edge may increase $a_D$ by $2/|D|$, which means $\floor{ 8 a_D }$ increases by at most one, while $M_D$ does not change.
    \end{itemize}
    In either case \AddAntiEdgeMatching is called before \RecolorDelete ends or \RecolorInsert reaches line \ref{line:begin-recolor-endpoint}. In order to end, \AddAntiEdgeMatching matches a new pair $\set{u,v}$ of previously unmatched vertices and calls \RecolorMatching{u,v} (line \ref{line:maintain-recolor-matching} in \cref{alg:recolor-matching}). Then \RecolorMatching{u,v} ends only after same-coloring $u$ and $v$ (line \ref{line:recolor-matching-uncolor} in \cref{alg:recolor-matching}), so the updated matching has size $\floor{ 8 \cdot a_D }$. The Dense Balance Invariant holds when the algorithms call \AddAntiEdgeMatching (because it held before the update and  uncoloring vertices does not affect it) and calls to \RecolorMatching and \RecolorDense preserve it (\cref{lem:recolor-denser-matching-invariants}).
\end{proof}

\subsection{Complexity of Recoloring the Denser Vertices}
\label{sec:denser}

We investigate the time complexity of our recoloring procedures for dense vertices and express it
in terms of the maximal size of a color class at the time of the update and the sparsity parameter $\zeta$. Observe that the probabilistic statements of \cref{lem:recolor-denser,lem:recolor-matching,lem:increase-matching-time} are only over the randomness sampled by the algorithm during the corresponding calls to \RecolorDense, \RecolorMatching and \AddAntiEdgeMatching.

\begin{lemma}
    \label{lem:recolor-denser}
    Suppose that the invariants hold before a call to $\RecolorDense{v}$ where $v\in D$, $\col[v] = \bot$ and $matched[v] = \bot$.
W.h.p. over the random colors it samples, the call ends in
    \[ O\paren*{ \max_{\chi} |\Phi[\chi]| \log^2 n + \zeta \log n } \]
    time.
\end{lemma}

\begin{proof}
    Suppose first that $v$ is an inlier, i.e., such that $a_v \leq 8 a_D$ and $e_v \leq 8 e_D$. We claim that every inlier has an available color in $L(D) \setminus \col[E(v)]$. Observe that this set is exactly $L(D) \cap L(v)$. Under the Matching Invariant, at least $|M_D| \geq \floor{ 8 a_D }$ colors are repeated in $D$ and since $v$ is uncolored, the accounting lemma (\cref{prop:accounting}) implies that
    \[
    |L(D) \cap L(v)|
    \geq |M_D| - a_v + 1
    \geq |M_D| - \floor{ 8 a_D } +1 \geq 1 \ ,
    \]
    where the second inequality uses the assumption that $a_v \leq 8 a_D$ and $a_v$ is an integer. To find this color, the algorithm loops over vertices of $E(v)$ and constructs the set $\col[E(v)]$ in $O(e_v \log n) \leq O(\zeta \log n)$ time. Then it loops over $L(D)$ and stops when it finds a color that is not in $\col[E(v)]$. If $L(D)$ contains more than $8e_D \geq e_v$ colors, the algorithm finds a color after $O(e_D \log n) \leq O( \zeta \log n)$ time. Otherwise, the algorithm goes over at most $8e_D \leq O(\zeta)$ colors; hence ends in $O(\zeta \log n)$ time.

    Suppose now that $v$ is an outlier. The algorithm samples colors in $[\Delta+1]$ until it finds one that is not used by (1) a matched vertex, (2) an external neighbor, or (3) an outlier. That is a color not in $M_D \cup \col[ E(v) \cup O_D ]$ (lines \ref{line:recolor-denser-if-free-outlier-1}-\ref{line:recolor-denser-if-free-outlier-2}). We claim that this set contains at most $\Delta/2$ colors.
    At most $2\eps\Delta$ colors are used by an external neighbors of $v$. By Markov's inequality, the number of vertices with $a_v > 8a_D$ or $e_v > 8e_D$ is at most $|C|/4 \leq (1/4 + \eps)\Delta$. Finally, observe that we always have $|M_D| \leq 8 \cdot 3\eps\Delta$. Indeed, \RecolorInsert and \RecolorDelete call \AddAntiEdgeMatching only when $|M_D| < 8 a_D \leq 24\eps\Delta$ because $a_D \leq 3\eps\Delta$. Calls to \RecolorMatching{u,v} where $u$ and $v$ are already matched do not increase the matching size.
    Adding these up, at most $(1/4 + 27\eps)\Delta \leq \Delta/2$ colors are blocked by $M_D \cup \col( E(v) \cup O_D )$ for $\eps \leq 1/108$.
    A random color in $[\Delta+1]$ therefore works \wp $1/2$, so the probability that the algorithm tries more than $c\log n$ random colors is $n^{-c}$. To verify if a random color works, the algorithm iterates through its color class in $O(\max_\chi |\Phi[\chi]|\log n)$ time. Then, there might be a additional call to \RecolorDense{u} where $u \in I_D$, which results in the claimed runtime.
\end{proof}

\RecolorMatching colors $u$ and $v$ the same way \RecolorDense colors an outlier. The complexity of recoloring an unmatched denser vertex (line \ref{line:recolor-matching-uncolor}) is upper bounded by \cref{lem:recolor-denser}. So we obtain the following lemma:
\begin{lemma}
    \label{lem:recolor-matching}
    Suppose all invariants hold before a call to \RecolorMatching{u,v}. W.h.p., the call ends after
    $O\paren*{ \max_{\chi} |\Phi[\chi]| \log^2 n + \zeta \log n }$
    time.
\end{lemma}

It remains to bound the time spent on fixing the matching size with \AddAntiEdgeMatching.
\begin{lemma}
    \label{lem:increase-matching-time}
    Consider a call to \AddAntiEdgeMatching{D} before which the Dense Balance Invariants holds but not the Matching Invariant. W.h.p., it ends in
    $O\paren*{ \max_{\chi} |\Phi[\chi]| \log^2 n + \zeta \log n }$
    time and $|M_D|$ has increased by one.
\end{lemma}
\begin{proof}
We may assume that $a_D > 1/8$ since otherwise the Matching Invariant holds even with an empty matching. In particular, $F_D$ is not empty.
If $|M_D| < \floor*{ 8 \cdot a_D }$, at least $|F_D|/2$ anti-edges of $F_D$ have both endpoints unmatched because $F_D$ contains $a_D |D|/2$ anti-edges while the current colorful matching is incident to at most $|M_D| \cdot 6\eps\Delta < 48\eps \cdot a_D|D| < |F_D|/2$ of them for, say, $\eps \leq 1/100$. The probability that \AddAntiEdgeMatching samples $c\log n$ anti-edges in $F_D$ without finding any with both endpoints available is thus $n^{-c}$. \cref{lem:recolor-matching} then implies the result.
\end{proof}

\subsection{Recoloring the Sparser Vertices}
\label{sec:sparser}

\FreshColoring ensures vertices of $S$ have sufficiently many colors in their palette at the beginning of the phase to maintain large palettes over entire phases. Importantly, vertices in $S$ conflict only with neighbors in $S$.
Recall that $\LS(v) = [\Delta + 1] \setminus \col[\NS(v)]$ is the palette of colors $v\in S$ can adopt with respect to $\col$.

\begin{lemma}
    \label{cor:recolor-sparse-time}
    Suppose \FreshColoring succeeded.
    A call to \RecolorSparse{v} where $\col[v]=\bot$ and $v\in S$ ends in
    \[
    O\paren*{
        \Delta / \zeta \cdot \max_{\chi} |\Phi[\chi]| \log^2 n +
        \zeta\log n
    }
    \] time with high probability over the random colors it samples.
\end{lemma}
\begin{proof}
    Call $R$ the set of colors used by at least two denser neighbors of $v$. The algorithm samples colors in $[\Delta+1]$ until it finds one in $\LS(v) \setminus R$ and we claim that this set contains at least $t$ colors at all time. It is easy to verify that $|R| \leq |L(v)|$, where $L(v) := [\Delta+1]\setminus\col(N(v))$ is the set of available colors, hence that $|\LS(v) \setminus R| \geq |R|$ because $L(v) \subseteq \LS(v) \setminus R$. So we may assume that $|R| \leq t$. Now, since \FreshColoring succeeded (\cref{prop:fresh-coloring}), we have $|\LS(v)| \geq 3t$ when the phase begins. As each update recolors at most one sparser vertex, after $i \leq t$ updates, we have that $|\LS(v) \setminus R| \geq 3t - i - |R| \geq t$.

    Each random color trial (line \ref{line:recol-sparse-sample} in \cref{alg:recolor-sparse}) recolors $v$ with probability at least \[ \frac{|\LS(v) \setminus R|}{\Delta+1} \geq \frac{ t }{\Delta+1} = \frac{ \gamma \cdot \zeta }{\Delta+1} \ . \] Hence, \whp, the algorithm samples $O(\Delta \log n / \zeta)$ colors before finding one that fits. To verify if a color $\chi$ is suitable for recoloring $v$, the algorithm iterates through vertices of $\Phi[\chi]$ and checks if they belong to $\NS(v)$ or $\ND(v)$ in $O(\log n)$ time. Overall, finding a suitable $\chi$ takes $O\paren*{ \Delta / \zeta \cdot \max_{\chi} |\Phi[\chi]| \log^2n }$ time. Then, if there is a $w \in \ND(v)$ with color $\chi$, \RecolorSparse makes one call to \RecolorDense{w} or \RecolorMatching{w, matched[w]} which recolors $w$ in $O(\zeta \log n + \max_\chi |\Phi[\chi]|\log^2n)$ time by \cref{lem:recolor-denser,lem:recolor-matching}.
\end{proof}

Overall, our algorithm is efficient only if we can maintain small enough color classes. The Dense Balance property ensures that denser vertices contribute $O(n/\Delta)$ to each color class. The contribution of vertices in $S$ is accounted for separately, using an argument over the randomness of the entire phase:
\begin{lemma}
    \label{lemma:dyn-sparse-balance}
    Suppose that \FreshColoring succeeded. With high probability \emph{over the randomness of the whole phase}, for every $\chi \in[\Delta+1]$, its color class contains
    \begin{equation}
        \label{eq:sparse-balance}
    | \Phi[\chi] \cap S |
    \leq O\paren*{ \frac{n}{\zeta} + \frac{t}{\zeta} + \log n }
    \end{equation}
    sparser vertices before every update of the phase.
\end{lemma}

\begin{proof}
    Fix $i \leq t$. We show that \cref{eq:sparse-balance} holds before the $i$-th update of this phase with high probability. Fix $\chi \in [\Delta + 1]$ and count the number of sparser vertices to join $\Phi[\chi]$ during this phase. Define $X_j$ to indicate if, during the $j$-th update, a sparser vertex joins $\Phi[\chi]$. Observe that, by the success of \FreshColoring, \RecolorSparse colors $v$ with a uniform color from a set of at least $\gamma\cdot\zeta$ colors, regardless of the recolorings that occurred earlier in the phase. Hence, the probability that it recolors $v$ with $\chi$ is $\Exp*{ X_j \given X_1, \ldots, X_{j-1} }
    \leq O(1/\zeta)
    $.
    By the Chernoff bound with stochastic domination (\cref{lem:chernoff-dom}), we have $\sum_{j \in[i]} X_j \leq O(i/\zeta) + O(\log n)$ with high probability. So before the $i$-th update, the set $\Phi[\chi]$ contains the $O\paren{ n /\zeta }$ vertices initially present after \FreshColoring (Sparse Balance property in \cref{prop:fresh-coloring}) and at most $O(i/\zeta + \log n)$ new ones. This holds for all $i \in [t]$ by union bound and thus implies the lemma.
\end{proof}

\begin{remark}
    \cref{eq:sparse-balance} holds with high probability over all the randomness used since the phase began, unlike \cref{lem:recolor-denser,lem:recolor-matching,cor:recolor-sparse-time}, which only consider the current update. Conditioning on \cref{eq:sparse-balance} would introduce a bias in \cref{lem:recolor-denser,lem:recolor-matching,cor:recolor-sparse-time}, which explains why we instead express runtimes in \cref{lem:recolor-denser,lem:recolor-matching,cor:recolor-sparse-time} as functions of the color class sizes. Since all events occur with high probability, they hold simultaneously by the union bound.
\end{remark}

\subsection{Proof of \cref{thm:main}}
\label{sec:proof-theorem}

Let $\delta$ be the universal constant from \cref{prop:refined-part} and define $\eps = 1/110$ and $\zeta = n^{2/3}$.
When $\Delta < n^{2/3}/(\eps^2 \delta)$, we use the naive algorithm that scans the neighborhood of the vertices after each update. When $\Delta \geq n^{2/3}/(\eps^2\delta)$, we use \cref{alg:insertion,alg:del}. We consider a fixed phase and show that the amortized update time is $\Ot{ n^{2/3} }$ with high probability. Recall that a phase is a sequence of $t = \gamma \cdot \zeta$ updates, where $\gamma$ is the universal constant described in \cref{prop:fresh-coloring}. Moreover, one should bear in mind that the algorithm of \cref{prop:refined-part} maintains an $(\eps,\zeta)$-sparser-denser decomposition with $O(\log n)$ update time such that the vertex partition is fixed throughout the phase.

\emph{Correctness \& Invariants.} \FreshColoring gives a total and proper $(\Delta+1)$-coloring and neither \RecolorInsert or \RecolorDelete introduce color conflicts. By induction, the invariants hold \whp before every update of the phase. Indeed, by \cref{prop:fresh-coloring}, if \FreshColoring ends, \whp, the invariants hold before the first update of the phase.
By \cref{lem:recolor-denser-matching-invariants,lem:matching-inv}, invariants are maintained by \RecolorDense, \RecolorMatching and \RecolorSparse, and thus are preserved after each update.

\emph{Efficiency.} By \cref{prop:fresh-coloring}, computing the fresh coloring takes $\Ot{ n + n^2 / \zeta }$ time. With high probability, \FreshColoring succeeds and henceforth we condition on its success. By \cref{cor:recolor-sparse-time,lem:recolor-denser,lem:recolor-matching}, \whp, each update ends in $\Ot*{ \frac{\Delta}{\zeta} \max_\chi |\Phi[\chi]| + \zeta }$ time. On the other hand, \whp, before every update and for every $\chi \in [\Delta + 1]$, we have that,
\[
    |\Phi[\chi]|
    = |\Phi[\chi] \cap S| + \sum_{i = 1}^r |D_i \cap \Phi[\chi]|
    \leq |\Phi[\chi] \cap S| + \frac{4n}{\Delta}
    \leq O\paren*{ \frac{n}{\zeta} + \log n }
\]
because, by the Dense Balance Invariant, each of the at most $2n / \Delta$ almost-cliques contributes at most two to $\Phi[\chi]$, and by
\cref{lemma:dyn-sparse-balance} for the sparse vertices. With the union bound, the bounds on the update time and the color class sizes hold with high probability.
Overall, \whp, the amortized update time during this phase is
\[
\Ot*{
    \frac{n + n^2 / \zeta}{t}
    + \frac{ n \Delta }{ \zeta^2 } + \zeta
}
= \Ot*{ \frac{n^2}{\zeta^2} + \zeta }
= \Ot*{ n^{2/3} }
\]
for $\zeta = n^{2/3} \leq O( \Delta ) \leq O( n )$ and $t = \gamma \cdot \zeta$.
\qed

\begin{remark}
    We hide three log-factors in our soft-Oh notation. We made no attempt at optimizing it but believe that using more sophisticated data-structures for sets and more careful concentration arguments, at least two of those could be avoided.
\end{remark} \section{Computing a Fresh Coloring}
\label{sec:fresh}

Recall that every $t$ updates, the algorithm recomputes a fresh coloring with nice properties. In this section, we provide and analyse this algorithm, thus prove the following proposition.

\PropFreshColoring*

\FreshColoring begins by uncoloring every vertex to start from a fresh state. It then uses a procedure called \OneShotColoring in which every vertex \emph{of $S$} tries a random color. It is a well-known fact that the remaining uncolored vertices have ``slack'' (\cref{thm:slack-gen}). The algorithm colors the remaining vertices in $S$ in a random order, which ensures that the total number of color trials is $\Ot{n}$ \cite{BRW24}.
Denser vertices do not participate in \OneShotColoring to ensure that the partial coloring it produced does not hinder the construction of a colorful matching. This is why the Sparse Slack guarantee is only in terms of $\LS(v)$ rather than $L(v) = [\Delta+1] \setminus \col[N(v)]$.
\FreshColoring colors almost-cliques sequentially, using the coloring primitives devised in \cref{sec:dynamic}.
\begin{algorithm}
\caption{Computing a Fresh Coloring\label{alg:fresh}}
\Fn{\FreshColoring}{
    \lFor*{each $v$}{\Color{v, $\bot$}} \tcp*[f]{uncolor all vertices first}

    \OneShotColoring    \tcp*[f]{generating slack}\;
    Sample a uniform permutation $\pi$ of $S$\;
    \lFor*{every uncolored $v\in S$ in the order of $\pi$}{ \RecolorSparse{v} \label{line:fresh-color-sparse} } \;

    \For{every $i = 1, 2, \ldots, r$}{\label{line:fresh-start-dense}
        \lFor{$\floor{8 a_{D_i}}$ times}{\AddAntiEdgeMatching{$D_i$} \label{line:fresh-matching}}
        \lFor*{every uncolored $v \in O_{D_i}$}{ \ColorDense{v} \label{line:fresh-color-outlier}}\;
        \lFor*{every uncolored $v \in I_{D_i}$}{ \ColorDense{v} \label{line:fresh-color-inlier}}\;
    }
}
\BlankLine\BlankLine

\Fn{\OneShotColoring}{
    \lFor*{every $v \in S$, with probability $1/8$}{ \Color{v, $\chi$} with a random $\chi \in [\Delta+1]$ \label{line:slackgen-random-col}}\;
    \For{every $\chi\in [\Delta+1]$ and every pair $u,v \in \Phi[\chi] \cap S$}{
        \lIf*{$\set{u,v} \in E$}{\Color($v$, $\bot$), \Color($u$, $\bot$)}
    }
}
\BlankLine\BlankLine

\Fn{\ColorDense{v}}{
    Sample $\chi \in L(C)$ \label{line:fresh-color-dense}\;
    \lForWithoutDo*{$u\in\Phi[\chi]$}{
        \lIf{$u\in N(v)$}{go to line \ref{line:fresh-color-dense}}
    }
    \Color{v, $\chi$}
}
\end{algorithm}

\paragraph{Generating Slack.}
To analyse the slack produced by \OneShotColoring, we use the following well-known results in the distributed graph literature (e.g., \cite[Lemma 3.3]{CLP20} or \cite[Lemma 6.1]{HKMT21} or even \cite{RM02}). Authors of \cite[Lemma A.1]{ACK19} and \cite{BRW24} use a version of this result for $\zeta \geq \Omega(\eps^2 \Delta)$.
\begin{proposition}
    \label{thm:slack-gen}
    Suppose $G$ is a graph of maximal degree $\Delta \gg 1$.
    If $\col$ is the coloring at the end of \OneShotColoring, then a $\zeta$-sparse vertex $v$ with $\zeta \gg 1$ has that\footnote{we did not attempt to optimize this constant}
    \[
        |L(v)|
        \geq [\#\text{uncolored vertices in $N(v)$}] + 2^{-22} \cdot \zeta
    \]
    with probability at least $1 - \exp \paren*{ - \Omega( \zeta) }$.
\end{proposition}

Through \cref{thm:slack-gen}, \OneShotColoring ensures the Sparse Slack property.

\begin{lemma}
    \label{lem:fresh-slack}
    Let $\zeta$ and $S, D_1, \ldots, D_r$ be as described by \cref{prop:fresh-coloring}.
    With high probability, we have that
    \begin{itemize}
    \item if \FreshColoring reaches line \ref{line:fresh-start-dense}, every vertex in $S$ is colored, the Sparse Slack and Sparse Balance properties hold; and
    \item before every call to \RecolorSparse (line \ref{line:fresh-color-sparse}), the Sparse Balance property holds.
    \end{itemize}
\end{lemma}

\begin{proof}
    \RecolorSparse{v} ends only after properly coloring $v$, so if the algorithm reaches line \ref{line:fresh-start-dense}, all vertices of $S$ must be properly colored. We first prove that the Sparse Slack properly holds \whp after \OneShotColoring, then the claim about the Sparse Balance. Both hold simultaneously \whp by union bound.

    We claim that when \OneShotColoring ends, \whp, every vertex $v\in S$ has
    \begin{equation}
        \label{eq:sparse-slack}
    |\LS(v)|
    \geq [\#\text{uncolored vertices in $\NS(v)$ }]
        + 4\gamma \cdot \zeta
        \quad\text{where}\quad
        \gamma = 2^{-25} \ .
    \end{equation}
    Indeed, consider a fixed vertex $v \in S$. By definition of $S$, the vertex $v$ is $\zeta$-sparse in $G$. Since we run \OneShotColoring in $G[S]$, let us argue that $v$ is $(\zeta/4)$-sparse in $G[S]$. We may assume that $v$ has degree at least $\Delta - \zeta$, since it otherwise always have $\zeta$ available colors. If $v$ has at least $\zeta/4$ neighbors in $V\setminus S$, then \cref{eq:sparse-slack} holds since $\LS$ ignores colors of neighbors in $V\setminus S$. Otherwise, $G[N(v)]$ contains at least $(\zeta/2)\Delta$ anti-edges, of which at most $(\zeta/4)\Delta$ are incident to $V \setminus S$, so $v$ is at least $(\zeta/4)$-sparse in $G[S]$. Since we assume that $\zeta \geq \Omega(\log n)$ for some large enough hidden constant, \cref{thm:slack-gen} applies to $v$ and, with probability at least $1 - \exp(-\Omega(\zeta)) \geq 1 - 1/\poly(n)$, we obtain that
    $
    |L(v)| \geq [\#\text{uncolored vertices in $\NS(v)$}] + 4\gamma \cdot \zeta
    $
    for $v$. This holds \whp for every vertex in $S$ by union bound over its vertices. This implies \cref{eq:sparse-slack} because after \OneShotColoring, only vertices of $S$ are colored, thus $L(v) = \LS(v)$. Calls to \RecolorSparse colors extends the coloring; thus, when the algorithm reaches line \ref{line:fresh-start-dense}, the Sparse Slack property holds.

    Denote by $v_1, \ldots, v_k$ the uncolored vertices of $S$ after \OneShotColoring in the order by which they are colored by \FreshColoring.
    Consider a fixed $\chi\in[\Delta+1]$ and let us now prove that before the $i$-th call to \RecolorSparse for fixed a $i\in[k]$, \whp over the randomness of sampled so far, we have $|\Phi[\chi]| \leq O(n / \zeta + \log n)$.
    During \OneShotColoring, a vertex may adopt the color it sampled and, by the classic Chernoff Bound, \whp, every color is sampled by at most $O(n/\Delta + \log n)$ vertices. Define $X_j$ the random variable indicating if $v_j$ is colored with $\chi$ by \RecolorSparse. \cref{eq:sparse-slack} implies \RecolorSparse colors $v$ with a color uniformly distributed in a set of size at least $4\gamma\cdot \zeta$. More formally, we have that $\Exp{ X_j \given X_1, \ldots, X_{j-1} } \leq 1/(4\gamma \zeta)$. So, we expect $O(i/\zeta) \leq O(n/\zeta)$ vertices to be colored with $\chi$ over the first $i$ calls to \RecolorSparse. The Chernoff Bound (\cref{lem:chernoff-dom}) shows that, \whp, when the algorithm makes its $i$-th call to \RecolorSparse, the color class contains $O(n/\Delta + n/\zeta + \log n) = O(n/\zeta + \log n)$ vertices. This holds \whp for every $\chi \in[\Delta+1]$ and $i\in[k]$ by union bound.
\end{proof}

As \cite[Appendix A.1]{BRW24} showed, this algorithm makes $\Ot{n}$ color trials (executions of line \ref{line:recol-sparse-sample} in \RecolorSparse) with high probability. We provide a proof of that fact in \cref{sec:proof-n-trials} for completeness. We obtain the following corollary by union bound: \OneShotColoring runs in $O(n + \Delta\max_\chi|\Phi[\chi]|^2) = \Ot{ n + n^2 / \Delta }$ and each of the $\Ot{n}$ color trials takes $O(n/\zeta + \log n)$ time by \cref{lem:fresh-slack}.
\begin{corollary}
    \label{cor:fresh-sparse}
    W.h.p., \FreshColoring reaches line \ref{line:fresh-start-dense} in $\Ot{ n + n^2 / \zeta }$ time, every vertex in $S$ is colored, the Sparse Slack and Sparse Balance properties hold.
\end{corollary}

\paragraph{Coloring Denser Vertices.}
Random color trials in $[\Delta+1]$ according to a random ordering could also be used for dense vertices, but as we need to ensure the Matching and Dense Balance invariants, we use a more tailord approach albeit similar.

\begin{lemma}
    \label{lem:fresh-dense}
    Let $\zeta$, $S, D_1, \ldots, D_r$ be as described by \cref{prop:fresh-coloring}. For any partial coloring outside of $D = D_i$ such that the Sparse Slack, Sparse Balance and Dense Balance hold, \whp, \FreshColoring extends the coloring to $D$ (lines \ref{line:fresh-matching} to \ref{line:fresh-color-inlier}) in $\Ot{ n + n\Delta/\zeta }$ time and such that $D$ contains a colorful matching of size $\floor{ 8 a_D }$ and the Dense Balance holds.
\end{lemma}

\begin{proof}
Since we begin with all vertices uncolored, \AddAntiEdgeMatching never makes a call to \RecolorDense. Using the same argument as \cref{lem:increase-matching-time}, it samples at most $O(\log n)$ anti-edges before finding one with both endpoints unmatched and tries $O(\log n)$ colors before finding one available for both endpoints. The algorithm therefore spends
\[
O(a_D) \times O\paren*{ \max_{\chi} |\Phi[\chi]| \log^2 n }
= \Ot*{\frac{n\Delta}{\zeta} }
\]
time computing a colorful matching (line \ref{line:fresh-matching}) because $8a_D \leq O(\Delta)$. Each call to \AddAntiEdgeMatching increases the size of the colorful matching by one, so the matching and Dense Balance invariants hold when \FreshColoring reaches line \ref{line:fresh-color-outlier}.

By Markov's inequality, $D$ contains at most $|D|/4$ outliers. Since all inliers $3|D|/4 \geq \Delta/2$ of $D$ are uncolored, the accounting lemma (\cref{prop:accounting}) shows that at least $\Delta / 2 - a_v \geq \Delta / 3$ colors of the clique palette are available for each outlier $v$. Hence, \ColorDense{v} where $v\in O_D$ ends in $O(\max_\chi |\Phi[\chi]|\log^2 n) = \Ot{ n/\zeta }$ time with high probability. So the coloring of outliers in $D$ takes up to $\Ot{ n\Delta/\zeta }$ time.

Suppose inliers are colored in the order $v_1, v_2, \ldots, v_k$. When the algorithm calls \ColorDense{$v_i$}, at least $k -i + 1$ remain uncolored in $D$. In particular, a random color is suitable to color $v_i$ with probability at least $(k - i + 1)/(16 c \cdot \zeta)$ where $c = \Theta(1/\eps^2)$ is a constant. Indeed, recall that an inlier has $a_v \leq 8a_D$ and $e_v \leq 8e_D$ while $D$ has that $a_D + e_D \leq c \cdot \zeta$ for some constant $c$ (\cref{def:refined-acd}). If $|L(D)| > 16c \cdot \zeta$, then $|L(D) \cap L(v)| = |L(D)| - e_v \geq |L(D)|/2$ and a random color in $L(D)$ belongs to $L(v)$ with probability at least $1/2$. Otherwise, the accounting lemma (\cref{prop:accounting}) implies that $|L(D) \cap L(v)| \geq k - i + 1$, which means that a random color in $L(D)$ is available to $v$ with probability at least $(k - i + 1)/|L(D)| \geq (k - i + 1)/(16c \cdot \zeta)$. In any case, after $O(\frac{\zeta\log n}{k - i + 1})$ color tries, \ColorDense finds a color for $v$ with high probability. Searching the color class for each color trial, the total time necessary to color inliers of $D$ is
\[
\sum_{i=1}^k O\paren*{ \frac{\zeta \log n}{k - i + 1} }
    \times O\paren*{ \max_{\chi} |\Phi[\chi]| \log n }
= O\paren*{ n \log^3 n }
\]
because $\max_\chi|\Phi[\chi]| \leq O(n/\zeta)$ from the sparse and Dense Balance and $k\leq 2\Delta$. So the time needed to color $D$ is $\Ot{ n + n\Delta/\zeta }$ as it takes $\Ot{ n\Delta/\zeta }$ time to color the matching and outliers and $\Ot{n}$ time to color inliers.
\end{proof}

\cref{prop:fresh-coloring} follows trivially from combining \cref{cor:fresh-sparse,lem:fresh-dense} for all $O(n/\Delta)$ almost-cliques.

\paragraph{Acknowledgment.}
Discussions with participants of Dagstuhl Seminar 24471, Graph Algorithms: Distributed Meets Dynamic, served as a key motivation for this work.

\bibliographystyle{alpha}
\bibliography{references}

\appendix

\section{Concentration Inequalities}
We use the Chernoff Bound to prove concentration of random variables around their expected values. To comply with dependencies between our random variables, we use a stronger form than the classic variant of the bound. We refer readers to \cite{DP09} or \cite[Section 1.10]{Doerr2020} or \cite[Corollaries 6 and 14]{KQ21}.

\begin{lemma}[Chernoff bounds with stochastic domination]\label{lem:chernoff-dom}
    Let $X_1, \ldots, X_n$ be random variables and $Y_i$ be any function of the first $i$ variables with value in $[0,1]$. Consider $Y = \sum_{i=1}^n Y_i$.
    If there exists $q_1, \ldots, q_n \in [0,1]$ such that $\Exp*{ Y_i \given X_1, \ldots, X_{i-1} } \le q_i$ for every $i\in [n]$, then for any $\delta>0$, 
    \begin{align}
      \label{eq:chernoffless}
      &\Prob*{ Y \ge (1+\delta) \mu }
      \le \exp\paren*{ - \frac{ \delta^2 }{ 2 + \delta } \mu }
      \quad\text{where}\quad
      \mu \geq \sum_{i=1}^n q_i\ .
    \shortintertext{Conversely, if there exists $q_1, \ldots, q_n \in [0,1]$ such that $\Exp*{ Y_i \given X_1, \ldots, X_{i-1} } \geq q_i$ for every $i\in [n]$, then for any $\delta \in (0,1)$,}
    \label{eq:chernoffmore}
      &\Prob*{ Y \le (1-\delta)\mu }
      \le \exp\paren*{ -\frac{\delta^2}{2}\mu}
      \quad\text{where}\quad
      \mu \leq \sum_{i=1}^n q_i\ .
    \end{align}
\end{lemma}

\section{Missing Proofs}
\label{sec:missing-proofs}

\subsection{Nearly Linearly Many Color Trials for Fresh Colorings}
\label{sec:proof-n-trials}

This is a streamlined adaptation of \cite[Appendix A.1]{BRW24}. 
It was already shown in \cite{FM24} that the \emph{expected} number of trials is $O(n\log n)$ --- even when vertices are not necessarily sparse. 
In fact, our proof does not rely on the sparsity or the slack provided by \OneShotColoring at all.

\begin{lemma}
    The number of color trials (line \ref{line:recol-sparse-sample} in \RecolorSparse) made by \FreshColoring is $O(n\log^2 n)$ with high probability.
\end{lemma}
\begin{proof}
    First, observe that a vertex with fewer than $\Delta/2$ neighbors in $S$ always has $\Delta/2$ colors available. Hence, \RecolorSparse colors such a vertex in $O(\log n)$ trials with high probability. We henceforth focus on vertices with at least $\Delta/2$ neighbors in $S$, which we refer to as high-degree vertices. Note that, \whp, every high-degree vertex has at least $\Delta/4$ uncolored neighbors after \OneShotColoring because vertices try a color only with probability $1/8$ (line \ref{line:slackgen-random-col}).

    Let $k = \Delta/(c\log n)$. We can see the process of sampling a uniform permutation of the vertices as the result of the following process: each vertex samples independently a uniform value in $[k]$; call $B_i$ the set of vertices that sampled $i$ and sample a uniform permutation of $B_i$; the resulting permutation is the one obtained by ordering vertices according to $B_1, B_2, \ldots, B_k$ first and then according to the internal order we sampled. Observe that by the Chernoff Bound, we have the following direct facts with high probability: every high-degree vertex has at least $\Delta/(8k) = (c/8)\log n$ uncolored neighbors in every $B_i$, and every $B_i$ contains at most $2n/k$ vertices. Now, for every $i\in \set{0,1,\ldots,\ceil{ \log k }}$, let $V_i = \bigcup_{j = 0}^{2^i - 1} B_{k-j}$. A high-degree vertex $v \in V_i \setminus V_{i-1}$ for $i>0$ has at least $2^{i-1}(c/8)\log n$ uncolored neighbors in $V_{i-1}$, which implies that it is still uncolored after $(\Delta + 1)/2^{i-1}$ color trials with probability at most
    \[
    \paren*{ 1 - \frac{2^{i-1} (c/8)\log n}{\Delta+1} }^{(\Delta + 1)/2^{i-1}}
    \leq \exp\paren*{ -\frac{2^{i-1} (c/8)\log n}{\Delta+1} \cdot \frac{\Delta + 1}{2^{i-1}} } \leq n^{-c/8} \ .
    \]
    Overall, the algorithm colors every vertex in $V_i \setminus V_{i-1}$ for $i > 0$ with
    \[
    |V_i| \cdot \frac{\Delta+1}{2^{i-1}}
    \leq 2^i \frac{2n}{k} \cdot \frac{\Delta+1}{2^{i-1}}
    \leq O(n\log n) \ .
    \]
    many trials.
    The algorithm therefore uses $O(n\log^2 n)$ color trials to color vertices in $B_1, \ldots, B_{k-1}$. Coloring a vertex in $B_k$ uses $O(\Delta\log n)$ color trials \whp, so the total number of color trials made by the algorithm is $O(n\log^2 n) + O(n\log n/\Delta) \cdot O(\Delta\log n) = O(n\log^2 n)$.
\end{proof}

\end{document}